\documentclass{llncs}
\pdfoutput=1
\usepackage{amsmath}
\usepackage{amssymb}
\usepackage{url}
\usepackage{stmaryrd}
\usepackage{paralist}

\newtheorem{notation}{Notation}

\newcommand{\mathify}[1]{\ifmmode{#1}\else\mbox{$#1$}\fi}

\newcommand{\N}{\mathbb{N}}
\newcommand{\Z}{\mathbb{Z}}
\newcommand{\R}{\mathbb{R}}

\newcommand{\defeq}{\;\mathrel{\colon\!\!\!\colon\!\!\!\!=\;}}

\renewcommand{\le}{\leqslant}

\newcommand{\lt}{<}

\renewcommand{\ge}{\geqslant}

\newcommand{\nin}{\not\in}

\newcommand{\set}[1]{\mathify{\left\{ #1 \right\}}}
\newcommand{\zfset}[2]{\mathify{\left\{ #1 \;\left|\; #2  \right.\right\}}}

\newcommand{\ord}[1]{\left[{#1}\right]}

\newcommand{\wehave}{.\;}

\def\cX{{\cal X}}

\newcommand{\hset}[1]{\mathify{\left\{\!\left| {#1} \right|\!\right\}}}
\newcommand{\hempty}{\emptyset}
\newcommand{\hsum}{\oplus}
\newcommand{\hmult}{\otimes}
\newcommand{\hdiff}{\ominus}
\newcommand{\join}{\obar}
\newcommand{\reduce}[1]{\mathify{\mathcal{R}({#1})}}

\newcommand{\pstar}{P_{\mathord{\ast}}}
\newcommand{\qstar}{Q_{\mathord{\ast}}}
\newcommand{\rstar}{R_{\mathord{\ast}}}
\newcommand{\collect}{\mathcal{C}}
\newcommand{\ordcollect}[1]{\mathcal{C}_{\ord{#1}}}
\DeclareMathOperator{\supp}{supp}

\begin{document}
\title{Symbolic Domain Decomposition}

\author{%
  Jacques Carette\inst{1} \and Alan P. Sexton\inst{2} \and Volker Sorge\inst{2} \and Stephen M. Watt\inst{3}}
\institute{%
  $^1${Department of Computing and Software}, {McMaster University} \\
  {\url{www.cas.mcmaster.ca/~carette}}\\
  $^2${School of Computer Science},  {University of Birmingham} \\
  {\url{www.cs.bham.ac.uk/~aps|~vxs}}\\
  $^3${Department of Computer Science},
  {University of Western Ontario}\\
  {\url{www.csd.uwo.ca/~watt}}
}
\maketitle

\begin{abstract}
  Decomposing the domain of a function into parts has many uses in mathematics.
  A domain may naturally be a union of pieces, a function may be defined by
  cases, or different boundary conditions may hold on different regions. For any
  particular problem the domain can be given explicitly, but when dealing with a
  family of problems given in terms of symbolic parameters, matters become more
  difficult. This article shows how hybrid sets, that is multisets allowing
  negative multiplicity, may be used to express symbolic domain decompositions
  in an efficient, elegant and uniform way, simplifying both computation and
  reasoning. We apply this theory to the arithmetic of piecewise functions and
  symbolic matrices and show how certain operations may be reduced from
  exponential to linear complexity.
\end{abstract}

\section{Introduction}
\label{sec:intro}

The goal of this paper is to develop general methods to work with domains having
symbolically defined parts. The \textit{raison d'\^etre} of symbolic
mathematical computation is to compute and reason about general expressions,
rather than working with particular values valid only at specific points.
Matters are simplest when variables range over a domain of interest and all
expressions are valid over the entire domain. Sometimes it is useful to perform
simplifications or other operations that are valid over part, but not all, of
the domain. In this situation, software systems may or may not record the
excluded region. But we are not always so fortunate to have this one-region
situation. More generally, the domain of interest may be made up of several
pieces with expressions taking different forms on different parts. Moreover, the
demarcation of the parts may be defined symbolically. This paper explores how to
represent such expressions concisely in a uniform way that simplifies
computation and reasoning.

When we do arithmetic with piecewise functions defined on explicit partitions,
we can do a mutual refinement of domain partitions to obtain regions that may
each be handled uniformly. When the partitions are defined symbolically,
however, we obtain a massive explosion of cases --- for $N$ binary operations on
functions of $k$ pieces there are $k^N$ potential regions. We say ``potential''
regions because, of this large number, not all of the regions are in fact
feasible. Furthermore, it is usually not possible to determine which regions are
feasible and which are not. For example, the sum 
$\sum_{i=1}^N f_i(x)$ of the functions  
\begin{equation*}
f_i(x) = \begin{cases} 0 & \text{for~} x < k_i, \\ A_i & \text{for~} x \ge k_i.\end{cases}
\end{equation*}
has $\sum_{i=0}^N N!/i!$ possible orderings of the $k_i$, with each ordering having
between 2 and $N+1$ regions. There is no ordering in which all $2^N$ regions are
realized.

We take the view that it is generally preferable to have a single compact
closed-form expression rather than a collection of cases, even if it means
introducing some new operations. For example, we are perfectly satisfied using
the Heaviside step function and giving the sum of the $f_i$ as $\sum_{i=1}^N A_i
H(x-k_i)$.

In this paper we show how hybrid sets, a variation on multisets allowing negative
multiplicities, enable us to write elegant closed form expressions of the form
we desire.  This use of hybrid sets also allows us to define a generalised
notion of partition, where symbolically defined parts are combined in more
useful ways than the usual set operations.  Our approach unifies and 
generalises a number of other techniques, such as the use of oriented regions 
for domains of integration.

Introducing new operators or generalising existing ones to write single
closed form expressions is more than just a cosmetic re-arrangement. 
It allows one to perform arithmetic and simplifications on whole expressions, 
and to reason about the expression and about the regions themselves.   
It is already customary to do this for certain operators.   
For example, by defining $\int_a^b f \mathrm{d}x = - \int_b^a f\mathrm d x$, it
becomes possible to write identities that hold universally. Then we have that,
independent of the relative order of $a$, $b$ and $c$, and subject to $f$ being
defined on the requisite domains,
\begin{equation}\label{identity:int}
\int_a^b f \mathrm d x = \int_a^c f \mathrm d x + \int_c^b f\mathrm d x.
\end{equation}
With a little work, we can also generalise the integral formula to
integrating over oriented subsets of $\R^n$.
Some authors similarly adjust the definitions of other operators to obtain 
universally true statements.  Similarly, Karr~\cite{KarrJACM} defines the
summation operator $\sum_{m \le i < n}$ so that 
$\sum_{m \le i < n} = - \sum_{n \le i < m}$ when $n < m$.
This allows equations such as the
following to hold for any ordering of $\ell$, $m$, $n$: 
\begin{equation*}
\sum_{m \le i < n} \bigl  ( g(i+1) - g(i) \bigr ) = g(n) - g(m), ~~~
\sum_{\ell \le i < n} f(i) = \sum_{\ell \le i < m} f(i) + \sum_{m \le i < n} f(i).
\end{equation*}

This paper formalises and extends these ideas in several ways, giving a
generalised framework for domain partitions and piecewise defined functions. We
first introduce some preliminaries and hybrid sets (\S2) and then their
generalisations to our notions of generalised partitions and hybrid functions
(\S3). We then present how we can decompose domains of hybrid functions to allow
for the combination of their symbolically defined pieces (\S4), before
presenting some applications (\S5) and discussing some concrete examples (\S6).

\section{Preliminaries}
\label{sec:prelim}
\subsection{Partitions and piecewise functions}
\label{sec:partitions}

The domain of definition of various mathematical objects (functions, 
vectors, matrices, sequences, etc) may naturally be decomposed into a
(disjoint) union of pieces where our object is then defined 
uniformly.
When the pieces are given as an explicit union of provably disjoint
sets which form a partition of the domain, the interpretation of a given
expression is reasonably straightforward.  More formally,

\begin{notation}
  We use the notation $\collect_{i \in I} X_i$, or, more briefly, $\collect_{I}
  X_i$ to describe a collection of elements $X_1, X_2, \dots$, indexed by a set
  $I$.  For $n\in\N$, we denote by $\ord{n}$ the set $\set{1,\ldots,n}$.
\end{notation}

\begin{definition}
  A \emph{partition} of a set $U$ is a collection $\collect_{I} P_i$ of
  pairwise disjoint sets such that $\bigcup_{I} P_i = U$.
\end{definition}
A partition, $\collect_I P_i$, induces a total function
$\cX:U\rightarrow I$ which gives the index of $\pstar$ where each $u\in U$ sits.
Piecewise expressions are then defined on top of a partition.
\begin{definition}
  A \emph{piecewise expression} over a set $U$ is a collection $\collect_I
  (P_i,e_i)$ where $\collect_I P_i$ is a partition of $U$ and each $e_i$ is an
  expression.
\end{definition}
Typically, each $e_i$ contains a distinguished variable $y$ which is interpreted
to range over $U$. 
\begin{definition}
  We say that $f:U\rightarrow S$ is a \emph{piecewise-defined function} if we
  have a collection $\collect_I(P_i,f_i)$ where $\collect_I P_i$ is a partition
  of $U$, $\forall i \in I \wehave f_i:P_i\rightarrow S$, and
$$\forall x:U\wehave f(x) = f_{\cX(x)}(x).$$
We call $\collect_I (P_i,f_i)$ the \emph{definition} of $f$, and each $f_i$ a 
\emph{piece} of $f$.
\end{definition}
It is important to note that piecewise-defined functions use their
argument in two different ways: once \emph{geometrically} by choosing
a set $P_i$ ``over'' which to work, and once \emph{analytically} to 
evaluate a function $f_i$.  The definitions above are straightforward
generalisations of those found in~\cite{Carette:2007:ISSAC}.

Lastly, we have that arithmetic operations on piecewise-defined functions
operate component-wise.  Note that we will silently use the convention 
that operations defined on the codomain of a function are lifted pointwise
to apply to functions, in other words $(f+g)(x) = f(x) + g(x)$.

\begin{proposition}\label{prop:base-arith}
  Let $f,g : U\rightarrow S$ be two piecewise-defined functions on the same
  partition $\collect_I P_i$ of $U$, with $\collect_I f_i$ (respectively
  $\collect_I g_i$) the collections of pieces of $f$ (resp. $g$). Further, let
  $\star : S\times S\rightarrow S$. Then 
  $\collect_I \left(f_i \star g_i\right)$ is the collection
  of pieces of $f \star g$ over the partition $\collect_I P_i$.
\end{proposition}

Note how the partition is entirely untouched.  This ``separation of 
concerns'' is what enables us to separate the issues of domain 
decompositions from arithmetic issues of piecewise-defined functions
(and expressions).  

To simplify our presentation, we introduce a domain restriction
operation and a join combinator on (partial) functions, to allow us a 
more syntactic method of ``building up'' piecewise functions.  These are
quite similar to Kahl's table composition 
combinators~\cite{Kahl-2003a}.

\begin{definition}
  \label{def:fn-restriction}
  The restriction $f^A$ of a function $f$ to a domain specified by a set $A$
  is
  \begin{equation*}
    f^A(x) \defeq 
    \begin{cases}
      f(x) & \text{if $x \in A$}\\
      \bot & \text{otherwise}
    \end{cases}
  \end{equation*}
\end{definition}

\begin{definition}
  The join, $f \join g$, of two (partial) functions $f$ and $g$, is defined as
  \begin{equation*}
    (f \join g)(x) \defeq
    \begin{cases}
      f(x) & \text{if f(x) is defined and g(x) is undefined}\\
      g(x) & \text{if g(x) is defined and f(x) is undefined}\\
      \bot & \text{otherwise}
    \end{cases}
  \end{equation*}
\end{definition}
This allows us to rewrite a piecewise-defined function $f$ defined by
$\collect_I (P_i,f_i)$, in terms of its pieces as $$f = f_1^{P_1} \join
f_2^{P_2} \join \ldots \join f_n^{P_n}.$$

But our goal is to work with piecewise-defined functions where we have a
\emph{symbolic} partition.  We need some new tools for this, which we will
develop in the next two sections.

\subsection{Hybrid Sets}
\label{sec:hybrid-sets}

We consider an extension of multisets, in which elements can occur multiple
times, to \emph{hybrid sets}, where the multiplicity of an element in a hybrid
set can range over all of $\Z$, instead of just $\N_0$. Thus a hybrid set, over
an underlying set $U$, is a mapping $U \rightarrow \Z$, i.e., it is an element
of $\Z^U$. We use the following definition, adapted
from~\cite{Loeb92-hybridsets}:

\begin{definition}
  Given a universe U, any function $H:U \rightarrow \Z$ is called a
  \emph{hybrid set}. 
\end{definition}

We can immediately define some useful vocabulary for working with hybrid
sets.
\begin{definition}
  The value of $H(x)$ is said to be the \emph{multiplicity} of the element $x$.
  If $H(x) \neq 0$, we say that $x$ is a member of $H$ and write $x \in H$;
  otherwise, we write $x \nin H$. The \emph{support} of a hybrid set is the
  (non-hybrid) subset $S$ of $U$ where $s\in S \iff s\in H$; we will denote the
  support of $H$ by $\supp H$. We (re)use $\hempty$ to (also) denote the empty
  hybrid set, i.e. the hybrid set for whom all elements have multiplicity $0$.
\end{definition}

\begin{notation}
  We use the notation $\hset{x_1^{m_1}, x_2^{m_2}, \dots}$ to describe the
  hybrid set containing elements $x_1$ with multiplicity $m_1$, $x_2$ with
  multiplicity $m_2$, etc. While our notation allows writing hybrid sets with
  multiple copies of the same element with different multiplicities, these
  denote the same hybrid set as that denoted by the normalised form with one
  copy of each element with a multiplicity which is the sum of the
  multiplicities of the copies of that element in the non-normalised form. Thus
  $\hset{a^2, b^1, a^{-3}, b^4} = \hset{a^{-1}, b^5}$.
\end{notation}

Set unions are usually defined by the boolean algebra structure (and more
specifically, via $\vee$) of the membership relation.  For hybrid set,
this is replaced by arithmetic over $\Z$.
\begin{definition}
  We define the sum, $A \hsum B$ of two hybrid sets $A$ and $B$ over a universe
  $U$, to be their pointwise sum. That is $(A \hsum B)(x) = A(x) + B(x)$ for all
  $x \in U$. We similarly define their difference, $A \hdiff B$ to be their
  pointwise difference, and their product, $A \hmult B$ to be their pointwise
  product. Let $\hdiff B$ denote $\hempty\hdiff B$.
\end{definition}
In other words, we do not use operations $A\cup B$, $A\cap B$ and
$A\setminus B$ for hybrid sets, but just $\hsum$, $\hdiff$ 
and $\hmult$.
We can
easily establish some identities such as
$(A \hsum B) \hdiff A = B$, $A\hdiff A = \hempty$ and 
$A \hsum (\hdiff B) = A \hdiff B$ as these follow directly from $\Z$.

Putting all of this together, we get:
\begin{proposition}
$\Z^U$ is a $\Z$-module. 
\end{proposition}
\begin{proof}
The abelian group structure is given by $(\Z^U,\hsum, \hdiff, \hempty)$, and
$\Z$ acts on hybrid sets by $n H = u \mapsto n\cdot H(u)$.
\end{proof}

We need two more technical definitions, which will be useful later.
\begin{definition}
We say that two hybrid sets $A$ and $B$ are \emph{disjoint} if
$A\hmult B = \emptyset$.
\end{definition}

\begin{definition}
  We call a hybrid set \emph{reducible} if all its members have multiplicity
  $1$. We define a reduction function, \reduce{\cdot}, on reducible
  hybrid sets that returns the (normal) set of members of the hybrid set.
\end{definition}

We should note that these hybrid sets (sometimes also called generalised
sets) have been studied before.  Hailperin 
\cite{Hailperin-Both-BoolesLogicAndProb}
makes the case that Boole~\cite{Boole1854-LawsOfThought} actually started
from hybrid sets for his algebraization of logic, but restricted himself
to nilpotent solutions of the resulting equations, which then correspond
closely to our modern notion of Boolean algebra.
Whitney wrote two nice 
papers~\cite{Whitney32-bullAMS:logical-expansion,Whitney33-annals:characteristic-fns}
taking up the theme of algebraising logic via characteristic functions.  He
does allow arbitrary multiplicities, and derives some nice normal forms
for certain kinds of partitions, in a way foreshadowing some of our own
results (see \S\ref{sec:gen-parts} and \S\ref{sec:hybrid-dom-decomp}).
Blizard~\cite{Blizard89-MultisetTheory} focuses on multisets (disallowing
negative multiplicities) but has an extensive bibliography of related
works, several of which being on (mechanised) theorem proving; he then
formalised sets with negative membership in~\cite{Blizard90-NegativeMembership}.
Blizard concentrates on concepts of union and intersection which closely
resemble those of normal set theory, although he does also define the
sum union (but not other related concepts).
Burgin~\cite{Burgin92-ConceptOfMultisetsInCybernetics} lists several
more works on hybrid sets, some reaching back to the early middle ages.
Syropoulos~\cite{Syropoulos01-MathematicsOfMultisets} gives a very readable
introduction to both multisets and hybrid sets.

\section{Generalisations}
\label{sec:generalisations}

We now revisit a few basic mathematical constructs and show how they may be
modified to work with hybrid sets. This will provide the machinery that we need 
for symbolic domain decomposition. First, we will examine the notion of a hybrid
function on a domain. We then show how sets and hybrid sets may be decomposed
using a notion of generalised partitions --- an extension of partitions to the
hybrid set case. We then address the practical issue of how to make two hybrid
partitions compatible by constructing a common refinement. Finally, when working
with functions defined over hybrid partitions, we need some way to compute
values. Over any given point in the domain, we need to know which functions must
be evaluated in computing the final value, which is rather complex for hybrid
functions over generalised partitions.  For this task, we introduce the notion
of pseudo-functions. When expressions on generalised partitions are evaluated,
these pseudo-functions avoid evaluating at places where the functions are
undefined or where the values are not needed. 
This allows us to deal with the cases, as in equation (\ref{identity:int}),
where component functions are not defined on some
parts of the domain decomposition but any application of the function 
in those places would anyway have multiplicity zero 
(\textit{i.e.} not be used).

\subsection{Functions of hybrid domain}
\label{sec:fn-hybrid}

It turns out that a useful definition of a ``function'' involving hybrid sets is
not entirely straightforward. Defining its graph is easiest. The underlying
intuition is that we capture the restriction of a function to a domain through
the multiplicities of the elements of the function graph in a hybrid set. The
hybrid set of a single element of the function graph for element $x$ in $U$ is
of the form $\hset{(x,f(x))^{1}}$. Therefore the scalar multiplication of that
set by the multiplicity of $x$ in $A$ will impose the appropriate restriction.
We use our function restriction notation of Def.~(\ref{def:fn-restriction}) only
for this hybrid version of function restriction henceforth.
\begin{definition}\label{def:hybrid-graph}
Let $A$ be a hybrid set over $U$, $B\subseteq U$, $S$ a set and 
$f:B\rightarrow S$ a (total-on-$B$) function.
A \emph{hybrid function} $f^A : U\times S \rightarrow \Z$ is defined by
$$f^A = \bigoplus_{x\in B} A(x) \hset{(x,f(x))^{1}} $$
\end{definition}
Note how the hybrid set $\Z$-module structure automatically takes care of
restricting the sum over the support of $A$.  Caution: some hybrid sets
do not form a hybrid function (for example $\hset{(1,1)^1, (1,2)^1}$
is not a hybrid function).

Our definitions work just as well with partial functions as with total
functions. But if for a hybrid function $f^F$, $f$ is undefined at some point in
the support of $F$, $f^F$ will not be defined at that point either. So, without
loss of generality, we can always restrict $F$ to where $f$ is defined. For the
remainder of this paper, we shall assume this, i.e.~whenever we write a hybrid
function $f^F$, $f$ is total over $\supp F$.

\begin{definition}
We call a hybrid function $f^H$ \emph{reducible} if the hybrid set $H$ is
reducible.  We extend $\reduce{\cdot}$ in this case by
$$\reduce{f^H}(x) = 
    \begin{cases}
      f(x) & \text{if $H(x)=1$}\\
      \bot & \text{if $H(x)=0$}
    \end{cases}$$
\end{definition}
We can generalise the join combinator to hybrid
sets.  This definition is quite central to ``making things work''.

\begin{definition}
  The join, $f^F \join g^G$, of two hybrid functions 
  $f^F$ and $g^G$ (with codomain $B$), 
  gives a \emph{hybrid relation}, a subset of
  $U\times B \times \Z$ given by
  $$ f^F \join g^G \defeq f^F \oplus g^G $$
\end{definition}
This is a rather ``dangerous'' definition, as it moves us from the land
of functions to that of relations.  In other words, it is quite possible
that $f^F\join g^G$ restricted to $U\times B$ is no longer the graph of a
function, but the graph of a relation.  But this extra generality will 
be quite useful for us, although we will have to prove that in the cases
which interest us, the resulting hybrid relations are in fact (reducible)
hybrid functions.

\begin{theorem}\label{hybrid-join}
Let $A$, $B$ be hybrid sets over $U$, $S$ an arbitrary set, and
$f:U\rightarrow S$ a total function.  Then
\begin{enumerate}
\item $\reduce{f^{\emptyset}}$ is the empty function,
\item $f^A \join f^A = f^{2A}$
\item $f^A \join f^B = f^{A\hsum B}$, and thus a hybrid function,
\item For $g:U\rightarrow S$ another total function, then
 $f^A \join g^B = (f \join g)^{A\hsum B}$ if and only if
$A\hsum B=\emptyset$ (where $f \join g$ is the join of regular functions).
\item Let $H_1,H_2$ be hybrid sets, with $\supp H_1$ and $\supp H_2$ disjoint,
$f_1:\supp H_1 \rightarrow S$ and $f_2:\supp H_2 \rightarrow S$, then
$f_1^{H_1} \join f_2^{H_2} = (f_1 \join f_2)^{H_1 \hsum H_2}$
\end{enumerate}
\end{theorem}
The proofs are omitted, and follow straightforwardly from the definitions.
Note the strong dichotomy between (3) and (4), which comes from the fact that
the non-hybrid $\join$ is designed to work with functions defined over
separate regions.

\subsection{Generalised Partitions}
\label{sec:gen-parts}

Theorem~\ref{hybrid-join} tells us that some collections of hybrid sets
are better than others.  Being disjoint is much too strong a property.
Nicely, for hybrid sets, partitions easily generalise in useful ways.

\begin{definition}
We define a generalised partition of a (hybrid) set, $P$, to be a 
finite collection of hybrid sets, $\ordcollect{n} P_i$, such that 
$P_1 \hsum P_2 \hsum \dots \hsum P_n = P$
\end{definition}
All set partitions of a set are also generalised partitions.  
Conversely, a generalised partition of a reducible set is a set partition
if and only if each generalised partition element is reducible.

\begin{remark}
We have lifted the \emph{disjointness} condition on 
partitions.  For something to be called a partition of $P$, it is
necessary that the result be equal to $P$.  Still, $P$ belongs to a 
larger universe $U$, and a generalised partition's pieces range over
$U$.  As long as, in the end, all elements of $U\setminus P$ have multiplicity
$0$, we get a generalised partition.  In this way, we have also lifted
the \emph{coverage} condition. 
\end{remark}

\begin{proposition}
For any generalised partition $\ordcollect{n} P_i$ of a hybrid set $P$ over $U$, arbitrary set
$S$, and any function $f:\supp P \rightarrow S$,
\[
f^P  =  f^{P_1}\join f^{P_2}\join \ldots \join f^{P_n} =  f^{P_1\hsum P_2\hsum \ldots \hsum P_n}
\]
is a hybrid function.
\end{proposition}
For brevity, we sometimes write $f^P$ for either of the right-hand
side expressions above.
When we want to join different functions over a partition and still get a 
function, we have to be careful and ensure we are joining ``compatible''
functions.

\begin{definition}
Let $P_1,P_2$ be a generalised partition of a hybrid set $P$ over $U$,
$S$ an arbitrary set and $f^{P_1}:P_1\rightarrow S$, $g^{P_2}:P_2\rightarrow S$
hybrid functions.  We say that $P_1,P_2$ is a \emph{compatible} partition for
$f,g$ if
$f(x)=g(x)$ for all $x\in \supp P_1 \cap \supp P_2$.
\end{definition}

\begin{theorem}
Using the same notation as above, $f^{P_1}\join g^{P_2}$ is a 
hybrid function if and only if $P_1,P_2$ is a compatible partition for
$f,g$.
\end{theorem}

It is important to note that although $\join$ is an associative, commutative
operation, the notion of compatibility, while commutative, does not lift
to a simple associative condition.

\begin{remark}
Some of our computations will purposefully use
\emph{incompatible} partitions.  Note that
$$(f^U\join g^U) \join g^{\hdiff U} = f^U\join (g^U\join g^{\hdiff U})
     = f^U \join g^{\emptyset} = f^U$$
but that $f^U\join g^U$ is in general a hybrid relation, yet the final
result is a hybrid function whenever $f^U$ is.  We will ``design'' our
hybrid partitions with this feature in mind.
\end{remark}

\subsection{Refinement}

To do arithmetic with hybrid functions,
we first need the notion of a refinement and a common refinement.
This is similar to the treatment of \cite{Carette:2007:ISSAC} for piecewise functions.
\begin{definition}
  A \emph{refinement} of a generalised partition $\collect_I P_i$ of $P$ is another
  generalised partition $\collect_J Q_j$ (of another hybrid set $Q$ not necessarily
  equal to $P$) such that for every $i \in I$ there exists a sub-collection $Q_{j_k}$
  of $Q_j$ such that $P_i = Q_{j_1} \oplus Q_{j_2} \oplus \ldots \oplus
  Q_{j_m}$. A \emph{common refinement} of a set of generalised partitions is a
  generalised partition that is simultaneously a refinement to every partition
  in the set.
\end{definition}
A refinement in this sense may well seem ``larger'' than the original
partition, as in the next example.

\begin{example}
Let the interval $P=[0,1]$, seen as $\hset{P^1}$,
and $I_1=[-1,0)$, $I_2=(1,2]$, $I_3=[-1,2]$, then
$Q = \hset{I_1^{-1}, I_2^{-1}, I_3^{1}}$ is
a refinement of $P$.
\end{example}

\begin{definition}  A refinement is called \emph{strict} if, for
each generalised partition being refined, the support of the associated
sub-collection is equal to the support of the generalised partition it refines.
\end{definition}

\begin{example}
$\hset{[0,1]^1},\hset{(1,2]^1},\hset{(2,3]^1}$ is a common strict refinement of
the two (trivial) hybrid partitions $\hset{[0,2]^1}$ and $\hset{(1,3]^1}$.
\end{example}

\subsection{Pseudo-functions and pseudo-relations}

As seen in the example above, a refinement may ``spill over'' the original
domain, so that if we look at a hybrid function $f^P$ where the underlying
$f$ is defined exactly on (the support of) $P$, $f^Q$ evaluated ``pointwise''
will not make sense.  Nevertheless, we want $f^P=f^Q$.  To achieve this,
we apply the lambda-lifting trick already used in
\cite{Carette:2007:ISSAC}.

\begin{definition}
Using the same notation as in Def.~\ref{def:hybrid-graph},
we define a \emph{pseudo-function} $\tilde{f}^A$ as
$$\tilde{f}^A = \bigoplus_{x\in B} A(x) \hset{(x,f)^{1}} $$
i.e. as a member of $U\times (U\rightarrow S) \rightarrow \Z$.
A \emph{pseudo-relation} is defined similarly.
The \emph{evaluation} of a pseudo-function (resp. relation) is defined
by mapping each point $(x,f)^k$ to $(x,f(x))^k$.
\end{definition}

The usefulness of pseudo-functions comes from the following property.
\begin{proposition}
For all refinements $Q$ of the generalised partition $P$,
$\tilde{f}^P = \tilde{f}^Q$.
\end{proposition}
In other words, even though the pieces of $Q$ might ``spill over'',
if we first ``simplify'' $\tilde{f}^Q$ by performing $\bigoplus_i Q_i$
to get $P$, we get a result $\tilde{f}^P$ which can then be safely
evaluated.  We will elide the $\tilde{\;}$ to lighten the notation whenever
this would not lead to confusion.

Another useful property of pseudo-functions is that in some cases we can
simplify them, regardless of what the underlying function does.
\begin{proposition}
$\tilde{f}^P \join \tilde{g}^Q \join \tilde{g}^{\hdiff Q} = \tilde{f}^P$
\end{proposition}
One of the chief advantages of pseudo-functions is that we can do some symbolic
manipulations of expressions in terms of these functions as if they were 
defined on a much larger domain, as long as the eventual term we evaluate
does not involve any of these ``virtual'' terms.

To aid in such computations, when we have pseudo-functions $\tilde{f}^P$
and $\tilde{g}^P$, with $f,g:\supp P\rightarrow S$, and some binary operation
$\star: S\times S\rightarrow S$, we will allow ourselves to write
expressions such as $\tilde{f}^P\star\tilde{g}^P$ in the induced term algebra
over pseudo-functions.  As usual, we lift evaluation pointwise,
$(\tilde{f}\star\tilde{g})(x) = f(x)\star g(x)$.  Furthermore we say that
$\tilde{f} = \tilde{g}$ over a set $\supp P$ whenever 
$\forall x\in\supp P.f(x) = g(x)$.  In other words, we use extensional equality
for the intensional terms $\tilde{f}$ and $\tilde{g}$.  This
means that properties like commutativity, associativity and having 
inverses lift to the term algebra.  As an example, we have that
\begin{proposition}
If $\star:S\times S\rightarrow S$ is associative and commutative then
$$\tilde{g}^Q \star \tilde{f}^P \star \tilde{g}^{\hdiff Q} = \tilde{f}^P$$
\end{proposition}

\section{Hybrid domain decomposition}
\label{sec:hybrid-dom-decomp}

We now have all the ideas necessary to decompose hybrid domains. An elegant
consequence of the formalism is that it allows us to use linear algebra to
construct the partitions that we require.

Let $\ordcollect{n} A_i$ and $\ordcollect{m} B_j$, be generalised
partitions of $U$. We want to find a generalised partition of $U$ that is a
common strict refinement of $A_i$ and $B_j$ and has minimal cardinality. The
cardinality restriction is to minimise the number of terms required for a
symbolic representation of the resulting domain decomposition. Thus we want to
choose a minimal generalised partition, $\collect_I P_i$ of $U$ such that, in the
$\Z$-module of hybrid sets and for some integers $a_{i,j}$, $b_{i,j}$ we
have $\bigoplus_i P_i = U$ and $\forall i : 1..n \wehave \bigoplus_j a_{i,j} P_j
= A_i$ and $\forall j : 1..m \wehave \bigoplus_i b_{i,j} P_i = B_j$.

Since this forms a system of $n+m+1$ simultaneous equations, of which only
$n+m-1$ can be independent, because $A$ and $B$ are, separately, partitions of
$U$, we need that number of independent variables to solve the system. Thus the
cardinality of the minimal partition is $n+m-1$ in the general case, and can
only be smaller in specific cases if there are some extra dependencies among
$U, A_1,\dots, A_{n-1}, B_1, \dots, B_{m-1}$.

This result generalises to a decomposition of $U$ into a minimal generalised
partition that is a common strict refinement of $r$ generalised partitions of
cardinality $n_1, \dots n_r$ respectively: The minimal partition required,
assuming full independence of the individual domain partitions, has cardinality
\begin{equation*}
  \left(\sum_{i=1}^r n_i\right) + 1 - r
\end{equation*}
If all the individual partitions have the same cardinality, $n$, this reduces to
$ r(n-1) + 1$.

We can automatically compute suitable minimal strict refinement partitions $U$
as follows. If we remove the equations for $A_n$ and $B_m$ from the system of
equations in order to get an independent set of simultaneous equations, we get
$n+m-1$ equations that can be written as a linear system in the $\Z$-module: 
\begin{equation*}
  \label{eq:indep-sim}
  C \cdot
  \begin{pmatrix}
    P_1 \\
    \vdots \\
    P_{n+m-1}
  \end{pmatrix}
  =
  \begin{pmatrix}
    U \\
    A_1 \\
    \vdots \\
    A_{n-1} \\
    B_1 \\
    \vdots \\
    B_{m-1}
   \end{pmatrix}
\quad \text{\normalsize where~}
  C = 
  \begin{pmatrix}
    1 & 1 & \dots & 1 \\
    a_{1,1} & a_{1,2} & \dots & a_{1,n+m-1} \\
    \vdots  &         &       & \vdots \\
    a_{n-1,1} & a_{n-1,2} & \dots & a_{n-1,n+m-1} \\
    b_{1,1} & b_{1,2} & \dots & b_{1,n+m-1} \\
    \vdots  &         &       & \vdots \\
    b_{m-1,1} & b_{m-1,2} & \dots & b_{m-1,n+m-1}
  \end{pmatrix}
\end{equation*}

Note that $C$ is an integer matrix. Further, the partition choice matrix, $C$,
must be invertible and $C^{-1}$ must be an integer matrix so that we obtain
integral partitions of each domain piece with respect to our partition $P$ of
$U$. An integer matrix is invertible and has an integer matrix inverse if and
only if it has a determinant of $+1$ or $-1$. Hence our problem of constructing
an appropriate partition reduces to choosing an integer matrix of the form of
$C$ such that its determinant is $\pm 1$. Finally, note that this directly
generalises to an arbitrary number of piecewise functions, each of an arbitrary
number of pieces.

If we restrict ourselves to triangular matrices, we can choose $C$ to be any
integer triangular matrix (upper because the first row of $C$ is all $1$s) for
which the product of the diagonal elements is $1$. Again, a simple way to do this
is to choose $C$ to be all $1$s along the top row, $1$s along the diagonal and
$0$ everywhere else. Another possibility is all $1$s in the whole upper
triangle.

For example, two suitable choice matrices with their inverses are
\begin{equation*}
  \begin{pmatrix}
    1      & \hdots & \hdots & \hdots & 1 \\
    0      & 1      & 0      & \hdots & 0 \\
    \vdots & \ddots & \ddots & \ddots & \vdots \\
    \vdots &        & \ddots & \ddots & 0\\
    0      & \hdots & \hdots & 0      & 1 
  \end{pmatrix}^{\!-1}\!\!\!\!=
  \begin{pmatrix}
    1      & -1     & \hdots & \hdots & -1\\
    0      & 1      & 0      & \hdots & 0 \\
    \vdots & \ddots & \ddots & \ddots & \vdots \\
    \vdots &        & \ddots & \ddots & 0\\
    0      & \hdots & \hdots & 0      & 1 
  \end{pmatrix}
\quad
  \begin{pmatrix}
    1      & \hdots & \hdots & 1 \\
    0      & \ddots &        & \vdots \\
    \vdots & \ddots & \ddots & \vdots \\
    0      & \hdots & 0      & 1 
  \end{pmatrix}^{\!-1}\!\!\!\!=
  \begin{pmatrix}
    1      & -1     & 0      & \hdots & 0 \\
    0      & \ddots & \ddots & \ddots & \vdots \\
    \vdots & \ddots & \ddots & \ddots & 0\\
    \vdots &        & \ddots & \ddots & -1\\
    0      & \hdots & \hdots & 0      & 1 
  \end{pmatrix}
\end{equation*}

\section{Applications}
\label{sec:applications}

\subsection{Arithmetic on Piecewise Functions}
\label{sec:app-piecewise}

We are now ready to generalise the arithmetic properties 
(see prop.~\ref{prop:base-arith}) of hybrid functions
and pseudo-functions.

\begin{proposition}\label{prop:distri1}
Let $\ordcollect{n} P_i$ be a partition of $P$, $f^P = f_1^{P_1}\join \ldots \join f_n^{P_n}$
and $g^P = g_1^{P_1}\join \ldots \join g_n^{P_n}$ be two hybrid functions on
$P\rightarrow S$.  Let $\star:(S\times S)\rightarrow S$,
then for all $x\in \supp P$,
$$ f^P(x) \star g^P(x) = (f_1(x) \star g_1(x))^{P_1} \join \ldots \join
    (f_n(x) \star g_n(x))^{P_n}.$$
\end{proposition}
Note how we can apply the above proposition to $f^F\star g^G$, by first 
restricting $F$ and $G$ to be over a common support (as 
$x\star\bot = \bot\star y=\bot$), then taking 
a common \emph{strict} refinement of (the restricted) $F$ and $G$.

We would like to lift this strictness condition.  We can almost do this
with pseudo-functions -- and with the help of a marked $\join$, we can.

\begin{definition}
We can \emph{mark} a $\join$ with a binary operation 
$\star:S\times S\rightarrow S$, denoted $\join^{\star}$.
We define evaluation of $\join^{\star}$ by
$$(\tilde{f}^F \join^{\star} \tilde{g}^G)(x) = 
    (F(x)+G(x))\hset{(x, (\tilde{f}\star \tilde{g})(x))^1}$$
\end{definition}
This operation clearly inherits properties of $\star$ like commutativity,
associativity and invertibility.  Unlike $\join$, $\join^{\star}$ will 
always result in a (pseudo) function.

\begin{proposition}\label{prop:distri2}
Let $\pstar=\ordcollect{n} P_i$ be a partition of $P$, 
$\qstar= \ordcollect{m} Q_j$ another partition of $P$, and
$\rstar= \collect_K R_k$ a common refinement of $\pstar$ and $\qstar$.
Let $\tilde{f}^P = f_1^{P_1}\join \ldots \join f_n^{P_n}$
and $\tilde{g}^P = g_1^{Q_1}\join \ldots \join g_m^{Q_m}$ be two \emph{pseudo}
functions with codomain $S$.  Let $\star:(S\times S)\rightarrow S$, be
associative and commutative,
then $\star$ distributes over the partition $\rstar$ in terms of $\join^{\star}$.
By the results of section~\ref{sec:hybrid-dom-decomp},
we can always choose $\rstar$ to be
\[P_1\hsum
\ldots\hsum P_{n-1}\hsum Q_1\hsum \ldots \hsum Q_{m-1}\hsum (U\hdiff (P_1\hsum
\ldots\hsum P_{n-1}\hsum Q_1\hsum \ldots \hsum Q_{m-1}))\]
\end{proposition}

\begin{example}
Let $A_1=[0,a), A_2=[0,1]\setminus A_1, B_1=[0,b), B_2=[0,1]\setminus B_1$,
all seen as hybrid sets.  Let
\begin{equation*}
f(x) = \begin{cases}
    2 & 0\leq x\lt a \\ 
    0 & a\leq x \lt 1
       \end{cases}
\qquad \text{and}\qquad
g(x) = \begin{cases}
    5 & 0\leq x \lt b \\
    7 & b\leq x \lt 1
       \end{cases}
\end{equation*}
We choose the hybrid (symbolic!) partition $A_1, B_1\hdiff A_1, B_2$, which
simultaneously refines both.  Then after a few computations we get
\begin{equation}\label{ex:mult0}
 f*g = \hset{2*5}^{A_1}\join^{*}\hset{2*5}^{B_1\hdiff A_1}
                       \join^{*}\hset{0*7}^{B_2}
\end{equation}
regardless of whether $a<b$ or $a\geq b$; in fact either (or both) could be
outside of $[0,1)$ and the result, interpreted as a hybrid function, are
still correct.  Moving from one choice of partition to another is done by 
undoing the
distribution, performing the change of partition, and using commutativity
and associativity to regroup like terms.  Note that we should have written
$2*5$ as $(x\mapsto 2)\tilde{*}(x\mapsto 5)$, but we chose the above for
greater clarity.
\end{example}
It should be very clear that expressions such as equation~\ref{ex:mult0}
are a formal representation of a piecewise function, and need to be
\emph{interpreted} properly in each context.

\subsection{Identities for invertible operators}

When the binary operation we use is invertible, we can perform the operations at
any time, as operands may later be removed from a cumulative result by applying
their inverses. This may lead to considerable efficiency improvements; even
though values and their inverses are cancelled in the calculation (leading to no
net effect) this may be more efficient than retaining an ``un-evaluable''
expression as a pseudo-function.

\begin{proposition}\label{prop:invert}
Let $f^P$ be a hybrid function over $S$ where $(S,\star)$ has the 
structure of an Abelian group (where we will use $e$ for the unit and $-$
for the inverse), then for all $x\in \supp P$,
$$(f^P \star f^{-P})(x) = (f^P \star (-f)^P)(x) = P(x)\hset{ (x,e)^1 }$$
where $-f$ denotes $x\mapsto{-f(x)}$.
\end{proposition}
Both equalities follow readily from the definitions.

\subsection{Identities for linear operators}

\begin{definition}
For a linear operator $L$, and $f^P$ a hybrid function, 
$$L(f^P) \defeq L( x\mapsto P(x)\cdot f(x) )$$
\end{definition}
Note $L(f^P)$ may not be defined even when $L(f)$ is.  If the multiplicity
function $P(x)$ is uniformly bounded, then it will exist.

\begin{proposition}
Let $f^P$ be a hybrid function, $\ordcollect{n} P_i$ a partition of $P$ such
that each $P_i(x)$ is uniformly bounded, then
$$L(f^P) = \sum_{i=1}^{n} L(f^{P_i})$$
\end{proposition}
The above is the fundamental reason why, under Karr's definition, the 
summation identities of the introduction hold.

\begin{corollary}
For all total functions $f:\Z\rightarrow G$ with $G$ an Abelian group,
and all $\ell,m,n \in \Z$,
$$
\sum_{\ell \le i < n} f(i) = \sum_{\ell \le i < m} f(i) + \sum_{m \le i < n} f(i).
$$
\end{corollary}

\section{Examples}
\label{sec:examples}

We present two examples of the application of hybrid functions to symbolic
computation problems. The first example is concerned with the arithmetic of
symbolic matrices, the second presents the idea of merging symbolic spline
functions.

\subsection{Matrix Addition}
\label{sec:app-matrices}

Earlier work~\cite{SeSoWa09issac,SeSoWa09calculemus} introduced 
The idea of support functions has been introduced previously
to represent symbolic matrices --- matrices given in terms
of symbolic regions with underspecified elements and symbolic dimensions --- and
defined arithmetic operations between them. The paper~\cite{SeSoWa09issac} 
presented a support function based on half-plane constraints that enables 
full arithmetic,
but that suffered from a combinatorial explosion in the number of terms needed
to express sum or product matrices. The paper~\cite{SeSoWa09calculemus} moved to a
support function based on interval addition that automatically dealt with
cancellation for negative intervals. While this avoided the combinatorial
explosion, the approach was restricted to certain types of regions and could not
be easily generalised to matrix multiplication. Hybrid functions and generalised
partitions solve both of these problems simultaneously. We demonstrate this with
the example of matrix addition of two $2\times2$ symbolic block matrices. Let
\begin{displaymath}
  M_1 = 
  \begin{pmatrix}
    A_1 & B_1 \\
    C_1 & D_1 \\
  \end{pmatrix},\qquad
  M_2 = 
  \begin{pmatrix}
    A_2 & B_2 \\
    C_2 & D_2 \\
  \end{pmatrix}
\end{displaymath}
where $M_1$ and $M_2$ are $n\times m$ matrices, $A_1$ and $A_2$ are of
dimensions $h_1 \times k_1$ and $h_2 \times k_2$ respectively,
$n,m,h_1,h_2,k_1,k_2\in \N_0$. 

Let $U = \zfset{(i,j)}{1 \le i \le n \land 1 \le j \le m}$ be
the set of all cell points in the matrices. We define the region occupied
by a matrix block similarly, and refer both to a matrix block and to the region
it occupies by the same name, relying on context to distinguish them.
We can thus write
\begin{equation}
  \label{eq:M_1_2}
  M_1 = A_1^{A_1} \join B_1^{B_1} \join C_1^{C_1} \join D_1^{D_1}, \qquad 
  M_2 = A_2^{A_2} \join B_2^{B_2} \join C_2^{C_2} \join D_2^{D_2}
\end{equation}
To calculate $M_1 + M_2$ we:
\begin{inparaenum}[(1)]
\item Choose a suitable generalised partition $\pstar$ of $U$.
\item Rewrite each block of each matrix into terms restricted to the chosen
  partition.
\item Substitute into the expressions for the matrices.
\item Add the two matrices region-wise.
\end{inparaenum}

As established in the section~\ref{sec:hybrid-dom-decomp}, the maximal number
of partitions required in our case is $4+4-1=7$. We therefore
choose $6$
independent regions to be $A_1, B_1, C_1, A_2, B_2,C_2$ and obtain the seventh,
$P_1$, by subtracting all other regions from the universe $U$,
\begin{equation}
P_1=U \hdiff (A_1
\hsum B_1 \hsum C_1 \hsum A_2 \hsum B_2 \hsum C_2).\label{eq:p-one}
\end{equation}
We can then express regions $D_1$ and $D_2$ in terms of $P_1$: $D_1 = U \hdiff
(A_1 \hsum B_1 \hsum C_1) = P_1 \hsum A_2 \hsum B_2 \hsum C_2$, and $D_2 = P_1
\hsum A_1 \hsum B_1 \hsum C_1$.

Rewriting $M_1,M_2$ from Eq.~(\ref{eq:M_1_2}), we get:
\begin{align*}
  M_1 &= A_1^{A_1} \join B_1^{B_1} \join C_1^{C_1} \join D_1^{P_1 \hsum A_2 \hsum B_2 \hsum C_2} \notag \\
  &= D_1^{P_1} \join A_1^{A_1} \join B_1^{B_1} \join C_1^{C_1} \join D_1^{A_2} \join D_1^{B_2} \join D_1^{C_2}  
\end{align*}
\begin{equation*}
  M_2 = D_2^{P_1} \join D_2^{A_1} \join D_2^{B_1} \join D_2^{C_1} \join A_2^{A_2} \join B_2^{B_2} \join C_2^{C_2}
\end{equation*}
This lets us express the sum of the two matrices as the following pseudo-function:
\def\joinplus{\join^{\scriptscriptstyle +}}
\begin{align}\label{eq:matrix-sum}
  M_1 + M_2 &{=}
  \left( D_1 + D_2\right)^{P_1}  \!\!\!\!& \joinplus 
  \left( A_1 + D_2\right)^{A_1}  \joinplus
  \left( B_1 + D_2\right)^{B_1}  \joinplus
  \left( C_1 + D_2\right)^{C_1}  \notag\\
  &&\joinplus \left( D_1 + A_2\right)^{A_2}  \joinplus  
  \left( D_1 + B_2\right)^{B_2}  \joinplus
  \left( D_1 + C_2\right)^{C_2} 
\end{align}
Observe that the seven terms of the hybrid function fully capture the result of
the matrix addition independent of the order of the symbolic dimensions
$h_1,h_2,k_1,k_2$ of the original blocks. We demonstrate this by evaluating the
function for a couple of concrete points in the sum matrix.

First let $h_1<h_2$ and choose a concrete value of a cell $(i,j)$ where
$h_1<i\leq h_2$ and $1\leq j<k_1,k_2$. The point should therefore be in
a region composed of elements from $B_1$ and $A_2$. Instantiating the 
multiplicities in equation~(\ref{eq:matrix-sum}) verifies this.  Observe that
indeed the only regions with multiplicity $1$ are $B_1 = \zfset{(i,j)}{h_1 \le
i \le n \land 1 \le j \le k_1}$ and $A_2 = \zfset{(i,j)}{1 \le i \le h_2 \land
1 \le j \le k_2}$ whereas the multiplicities for $A_1,B_2,C_1,C_2$ are all $0$.
Furthermore, we can
compute the multiplicity for $P_1$ using equation~(\ref{eq:p-one}). Since the
multiplicity of the universe $U$ is always $1$ --- every element is in this
partition --- we get $1-(0+1+0+1+0+0)=-1$.  This then yields the computation
below, which confirms that our element is indeed in the anticipated region
(where we write region-wise sets $\hset{(D_1+D_2)^{-1}}$ as
$(D_1+D_2)^{-1}$ to alleviate notation)
\[
  \left( D_1 + D_2\right)^{-1} \joinplus \left( B_1 + D_2\right)^1  \joinplus \left( D_1 + A_2\right)^{1} = B_1^1\joinplus A_2^1 = B_1 + A_2 
\]

Now assume that the order of the symbolic dimensions for blocks $A_1,A_2$ is
reversed and we have $h_2<h_1$. If we now compute the value of a cell with
$(i,j)$ with $h_2<i\leq h_1$ and $1\leq j<k_1,k_2$, we get $0$ multiplicity for
regions $B_1,A_2$, but instead multiplicity $1$ for $A_1,B_2$. Again the
multiplicity for $P_1$ is $-1$ and equation~(\ref{eq:matrix-sum}) will yield
that our cell is in the $A_1+B_2$ region.

As a final example we compute cell $(i,j)$ with $h_1,h_2<i\leq n$ and
$k_1,k_2<j\leq m$, which is a point from $D_1$ and $D_2$. This time
equation~(\ref{eq:matrix-sum}) simplifies even more quickly, as the
multiplicities for $A_1,A_2,B_1,B_2,C_1,C_2$ are all $0$ and we only have to
consider the multiplicity for $P_1$ which is $1$ with the same considerations as
above, yielding $D_1+D_2$ as the only term that does not vanish.

\subsection{Symbolic Spline Interpolation}
\label{sec:app-splines}

\def\smerge{\Join}

Numerical computation and manipulation of splines is well understood. Here we
show how hybrid domain decomposition can be used to support the previously
unexplored symbolic manipulation of splines.

Let $a= x_0 < x_1 < \cdots < x_{n-1} < x_n=b$ be a partition of the interval
$\left[a,b\right]\subset\R$. We call the $x_i$ knots and assume that for each
knot we have a corresponding value $y_i$. Then we can define a spline
function $S$ over $\left[a,b\right]$ piecewise as
\begin{equation}
S(x) = \left\{\begin{matrix} 
             S_0(x) & x \in [x_0, x_1] \\ 
\vdots & \vdots \\ S_{n-1}(x) & x \in [x_{n-1},
    x_n] \end{matrix}\right.\label{eq:spline}
\end{equation}
While spline interpolation is traditionally defined for numerical values of the
$x_i,y_i$ pairs, we will now define a symbolic spline function. Let
$a=c_0<c_1<\cdots< c_{n-1}<c_n=b$ be symbolic or ``abstract'' knots with
associated symbolic values $d_0,d_1,\ldots,d_{n}$, where a $d_i$ is generally
given as a function in $c_i$.  We define spline segments $S_{c_i,c_{i+1}}(x)$
for $i=0, \ldots, n-1$. That is, the segments are parameterised with respect to
two knots and their values.  Let $P=P_1\hsum \ldots \hsum P_{n}$ with
$P_i=[c_{i-1},c_{i}]$ be a generalised partition. We then define a symbolic
spline function $S(x)= S^{P_1}_{c_0,c_{1}}(x) \join \cdots \join
S_{c_{n-1},c_{n}}(x)^{P_{n}}$.  For clarity we often omit the $(x)$ part of the
term.

Define the merge of two spline segments $S_{a,b}^P\smerge
S_{a',b'}^Q$ to be $S_{max(a,a'),min(b,b')}^{P\hmult Q}$. Clearly this merge will
be empty if the two segments do not overlap, otherwise it will be the smallest
possible spline for the overlapping interval of the segments. 

Now let $P_1\hsum\cdots\hsum P_n$ and $Q_1\hsum\cdots\hsum Q_m$ be two
generalised partitions of the interval $[a,b]$ and let $S= S^{P_1}_{c_0,c_{1}}
\join \cdots \join S_{c_{n-1},c_{n}}^{P_{n}}$ and $T= T^{Q_1}_{d_0,d_{1}} \join
\cdots \join T_{d_{m-1},d_{m}}^{Q_{m}}$ be two symbolic splines. Observe that
the $d_i$ here are knots and not knot values. We can then define the merge
$S\smerge T$ as a binary operation on two hybrid functions as given above in
proposition~\ref{prop:distri2}.

We observe the merge operation using a simple example. Let $P=P_1\hsum P_2$ and
$Q=Q_1\hsum Q_2$ be the generalised partitions $a<c<b$ and $a<d<b$ of our
universe $[a,b]$, respectively. Let $S= S_{a,c}^{P_1} \join S_{c,b}^{P_2}$ and
$T= T_{a,d}^{Q_1} \join T_{d,b}^{Q_{2}}$ be two symbolic splines. We choose a
common refinement as $P_1,Q_1,R=U\hdiff(P_1\hsum Q_1)$. We can then write $S=
S_{a,c}^{P_1} \join S_{c,b}^{R\hsum Q_1}=S_{a,c}^{P_1}\join S_{c,b}^{R} \join
S_{c,b}^{Q_1}$ and similarly $T= T_{a,d}^{Q_1} \join T_{d,b}^{P_{1}} \join
T_{d,b}^{R}$. When we merge both symbolic splines we get
\begin{equation}
  \label{eq:spline-ex}
  S\smerge T=
  (S_{a,c}\smerge T_{d,b})^{P_1}\join^{\smerge} (S_{c,b}\smerge
  T_{a,d})^{Q_1}\join^{\smerge} (S_{c,b}\smerge T_{d,b})^{R}
\end{equation}
If we now fix the order of our symbolic knots to be $a<c<d<b$ we can evaluate
the three components of our spline. First let $a\leq x\leq c$, which means
$P_1=Q_1=1$ and $R=-1$ and~(\ref{eq:spline-ex}) evaluates to $S_{a,c}\smerge
T_{a,d}$ which yields a spline between knots $a,c$. Similarly the other two
segments evaluate $S_{c,b}\smerge T_{a,d}$ and $S_{c,b}\smerge T_{d,b}$, which
yields splines for the intervals $[c,d]$ and $[d,b]$, respectively.

\section{Conclusion}
\label{sec:conc}

We have presented a framework of generalised partitions and domain decomposition
based on hybrid sets. This framework has a number of pleasing properties and
unifies a number of \textit{ad hoc} notions in common use. More importantly for
our purposes, this representation allows easy manipulation of and reasoning
about partitions whose pieces are defined symbolically.
These specialise correctly for all choices of
parameters by negative and positive multiplicities cancelling as needed.
The representation (see for example equation~(\ref{ex:mult0})) needs to be carefully interpreted, as
out of context simplification can result in meaningless results.  But if the
rules we lay out are followed, our compact representation effectively
allows one to compute with very general piecewise-defined functions.

Although a number of previous authors have studied hybrid sets, our 
study of functions over hybrid sets, generalised partitions, the superposition
$\join$ and marked superposition $\join^{\star}$ operators appear to all be
new.  Our applications to computation and reasoning are certainly new.

There remain several intriguing directions for future work. These include normal
forms for piecewise functions defined on hybrid partitions, simplification of
intermediate expressions involving linear operators or inverse elements and
creating partition schemes from algebraic specialisation properties,
\textit{e.g.} in Cylindrical Algebraic Decomposition. We have implemented a
prototype Maple package for hybrid sets and hybrid functions, but a more general
implementation would be of interest.

\bibliographystyle{plain}
\bibliography{calculemus10}

\end{document}